\documentclass[11pt,twoside]{article} 
\usepackage{acta-info}
\usepackage{euler}

\usepackage{epsfig}

\newcommand{\vol}{\textbf{3,} 2 (2011) 224--229}   
\newcommand{\amss}{\amssubj{%
68P10
}}     
\newcommand{\CRs}{\CRsubj{%
F.2.2
}}   
\newcommand{\keyw}{\keywords{%
search, lies
}}

\newtheorem{thm}{Theorem}
\newtheorem{cor}[thm]{Corollary}

\newtheorem{claim}[thm]{Claim}
\begin{document}

\title{Lower bounds for finding the maximum\newline and minimum elements with k lies}

\maketitle

\oneauthor{%
\href{http://www.cs.elte.hu/~dom}{D\"om\"ot\"or P\'ALV\"OLGYI} 
}{%
\href{http://www.elte.hu/en}{Lor\'and E\"otv\"os University Budapest}\\ \href{http://www.cs.elte.hu/index.html?lang=en}{Institute of Mathematics}
}{%
 \href{mailto:dom@cs.elte.hu}{dom@cs.elte.hu}
}

\short{D. P\'alv\"olgyi}{Finding the maximum and minimum elements with k lies}

\begin{abstract}
In this paper we deal with the problem of finding the smallest and
the largest elements of a totally ordered set of size $n$ using pairwise
comparisons if $k$ of the comparisons might be erroneous where $k$ is a fixed constant.
We prove that at least $(k+1.5)n+\Theta(k)$ comparisons are needed in the worst case thus disproving the conjecture that $(k+1+\epsilon)n$ comparisons are enough.
\end{abstract}

\section{Introduction}
Search problems with lies have been
studied in many different settings (see surveys Deppe \cite{D} and Pelc \cite{P}).
 In this paper we deal with the model when a fixed number, $k$,
of the answers may be false, which we call lies.
There are also several models depending on what kind of questions are allowed as well, the most famous being the R\'enyi-Ulam game.
In this paper we deal with the case when we are given $n$ different elements and we can use pairwise comparisons to decide which element is bigger from the two.

The problem of finding the maximum (or the minimum) element with $k$ lies was first solved by Ravikumar et al. \cite{RGL}. They have shown that $(k+1)n-1$ comparisons are necessary and sufficient.
The topic of this paper is finding the maximum and the minimum.
If all answers have to be correct then the minimum number of
comparisons needed is $\lceil \frac{3n}{2}\rceil -2$ (see
\cite{Po}).
Aigner in \cite{A} proved that $(k+\Theta(\sqrt k))n+\Theta(k)$ comparisons are always sufficient\footnote{He also obtained asymptotically tight results in another model.}.
It was proved by Gerbner et al. \cite{mienk} that if $k=1$, then $\frac{87n}{32}+\Theta(1)$ comparisons are necessary and sufficient. We also made the conjecture that for general $k$, there is an algorithm using only $(k+1+c_k)n$ comparisons where $c_k$ tends to $0$ as $k$ tends to infinity.
Hoffmann et al. \cite{HMOZ} showed that $(k+1+C)n+O(k^3)$ comparisons are sufficient for some absolute constant $C$ (whose value is less than $10$ but no attempts to optimize it were made yet).
Until now the best lower bound on $c_k$ was $\Omega((1+\sqrt 2)^{-k})$ by Aigner \cite{A}.
The main result of this paper is the following theorem.

\begin{thm}\label{main} At least $\left\lceil(k+1.5)(n-1)-0.5\right\rceil=(k+1.5)n+\Theta(k)$ comparisons are needed in the worst case to find the largest and the smallest element if there might be $k$ erroneous answers. 
\end{thm}

This bound is tight for $k=0$ (see Theorem \ref{max}) but not for $k=1$ as shown in \cite{mienk} and using a slightly more involved argument than the one presented here it is easy to see that the bound can be simply improved for any $k\ge 1$. The reason why the theorem is presented in this ``weak'' form is that it already disproves the conjecture and the argument is simply, yet gives a perfectly matching bound for $k=0$. To find a stronger version would involve a thorough case analysis, similar to the one in \cite{mienk} and improving the constant a bit is uninteresting at the moment. It would be more interesting to study the behavior of $c_k$ in future works. Now we know that $1.5\le c_k\le C\sim 10$. But is $c_k$ monotonously increasing as $k$ grows? This would imply, of course, the existence of a limit, which is likely to exist. 

The rest of the paper is organized as follows. In Section $2$ we develop a method to increase the lower bound by $k$ for many search problems and give proofs using it for some known results. In Section $3$ we prove our main result, Theorem \ref{main}.

\section{k more questions}
In this section, as a warm-up, we prove a very general result that holds for all search problems and generally gives an additional constant to the lower bounds that are proved using a consistent adversary.

\begin{claim}\label{kmore} Suppose we have a search problem where we want to determine the value of some function $f$ using (not necessarily yes-no) questions from a family of allowed questions.
The answers are given by an adversary who can lie at most $k$ times.
Suppose that we have already asked some questions and the answers we got are consistent, i.e. it is possible that none of them is a lie.
If we do not yet know the value of $f$, then we need at least $k+1$ further questions to determine it.
\end{claim}

This claim has an immediate, quite weak corollary.

\begin{cor} If there is a search problem as in Claim \ref{kmore} with a non-trivial $f$, then we need at least $2k+1$ questions to determine $f$.
\end{cor}

Although it is not too standard, we first give a proof of the Corollary, as it is a simplified version of the proof of the Claim.

\medskip
\begin{proof} Take two possible elements of the universe, $x$ and $y$, for which $f(x)\ne f(y)$. The adversary can answer the first $k$ questions according to $x$ and the next $k$ questions according to $y$, thus after $2k$ questions both are still possible. 

$  $ \qquad $ $ \end{proof}

\medskip
\noindent\textbf{Proof of Claim \ref{kmore}.} Suppose we have already asked some consistent questions, i.e. there is an $x$ such that they are all true for $x$. However, if we do not yet know $f$, there is a $y$ for which at most $k$ of these questions would be false, such that $f(x)\ne f(y)$. We can answer the next $k$ questions according to $y$. \hfill $\Box$

\medskip To show the power of this simple claim, let us prove the following theorem.

\begin{thm}[Ravikumar et al. \cite{RGL}]\label{max}  To find the maximum among $n$ elements using comparisons of which $k$ might be incorrect, we need $(k+1)n-1$ comparisons in the worst case.
\end{thm}

\begin{proof} The upper bound follows from using any tournament scheme and comparing any two elements until one of them is bigger than the other $k+1$ times. This is $(k+1)(n-1)$ plus the possible $k$ lies that might prolong our search.

To prove the lower bound, answer the first $(k+1)(n-1)-1$ questions consistently. Now we have an element that was always bigger, and another that was the smaller one at most $k$ times, thus the conditions of Claim \ref{kmore} are satisfied, so we need $k+1$ more questions.
\end{proof}

\newpage
\section{Proof of Theorem 1}
We start with defining some standard terminology.
Define the actual {\em comparison graph} as a directed graph whose vertices are the elements and it has an edge for every comparison between the compared elements, directed from the bigger towards the smaller. We say that the comparison graph is {\em consistent} if there is no directed cycle in the comparison graph. In this case any vertex with in-degree at most $k$ can still be the maximum element and any vertex with out-degree at most $k$ can still be the minimum element.
We also denote the comparison graph after the first $t$ questions by $G_t$.
So if there are no lies among the first $t$ answers, then they are necessarily consistent and there is no directed cycle in $G_t$.

Now we prove Theorem 1, which states that $\left\lceil(k+1.5)(n-1)-0.5\right\rceil$ comparisons are needed to find the largest and the smallest element if there might be $k$ erroneous answers.

\medskip\noindent
\textbf{Proof of Theorem 1.} We have to give an adversary argument, i.e., for every possible comparing algorithm, we have to give answers such that it is not possible to determine with less than $(k+1.5)(n-2)+1$ questions the maximum and the minimum.
Our answers will be always consistent, i.e., that there will be no directed cycle in the comparison graph.

First, we suppose that $n$ is even and the (undirected) edges of $G_{n/2}$ (the graph of the first $n/2$ questions) form a perfect matching, i.e., every element is compared exactly once during the first $n/2$ comparisons. Denote the set of elements that were bigger in their first comparison by $TOP$ and the ones that were smaller by $BOTTOM$. Whenever in the future an element from $TOP$ is compared to an element from $BOTTOM$, we always answer that the one from $TOP$ is bigger. This way the problem reduces to finding the maximum from $n/2$ elements and the minimum from $n/2$ other elements. Every vertex but one from $TOP$ must have in-degree at least $k+1$ at the end and, similarly, every vertex but one from $BOTTOM$ must have out-degree at least $k+1$ at the end. Therefore after $n/2+2(k+1)(n/2-1)-1$ comparisons we still cannot know both the maximum and the minimum, and the answers we got are all consistent, thus we need $k+1$ more questions because of Claim \ref{kmore}. This implies that at least $(k+1.5)(n-1)-0.5$ comparisons are needed in the worst case.

In general, define the sets $TOP$ and $BOTTOM$ to be empty at the beginning and whenever an element is first compared, put it to $TOP$ if it is bigger and to $BOTTOM$ if it is smaller than the element it is compared to. Whenever we compare and element from $TOP$ with an element from $BOTTOM$, always the $TOP$ one will be bigger, so the maximum will be in $TOP$ and the minimum in $BOTTOM$.
At the end of the algorithm every element must be assigned to $TOP$ or $BOTTOM$. Denote the number of elements that are put to $TOP$ by $n_1$ and the number of the ones that are put to $BOTTOM$ by $n_2$ (so we have $n_1+n_2=n$). It is clear that there are at least $\left\lceil n/2\right\rceil$ questions that compare at least one element that was not compared before. Also note, that if we compare two elements one of which is not in $TOP$, then the in-degree of the vertices in $TOP$ will not increase. Therefore we need at least $(k+1)(n_1-1)$ comparisons inside $TOP$. We similarly need at least $(k+1)(n_2-1)$ comparisons inside $BOTTOM$.
Therefore after $\left\lceil n/2\right\rceil+(k+1)(n-2)-1$ comparisons we still cannot know both the maximum and the minimum, and the answers we got are all consistent, thus we can apply Claim \ref{kmore}. This implies that at least $\left\lceil(k+1.5)(n-1)-0.5\right\rceil$ comparisons are needed in the worst case. Note that this equals $\left\lceil(k+1.5)n\right\rceil-k-2$, which for $k=0$ is $\left\lceil 3n/2\right\rceil$, matching the best algorithm and the result of \cite{Po}. \hfill $\Box$



\subsection*{Acknowledgement}
I would like to thank the members of Gyula's search seminar, especially Dani and Keszegh, to listen to my attempts to prove the conjecture that I eventually ended up disproving.

The European Union and the European Social Fund have provided financial support to the project under the grant agreement no. T\'AMOP 4.2.1/B-09/1/KMR-2010-0003.


\bigskip
\rightline{\emph{Received: October 15, 2011 {\tiny \raisebox{2pt}{$\bullet$\!}} Revised: November 10, 2011}} 


\begin{thebibliography}{99}\fontsize{10}{0}
\bibitem{A}
M. Aigner, Finding the maximum and the minimum,
\href{http://www.elsevier.com/wps/find/journaldescription.cws_home/505609/description}{\textit{Discrete}} \emph{Appl. Math.} \textbf{74,} 1 (1997) 1--12.

\bibitem{D}
C. \href{http://www.math.uni-bielefeld.de/ahlswede/papers/deppe.html}{Deppe}, Coding with feedback and searching with lies, in:
\emph{Entropy, Search, Complexity}, Bolyai Society \href{http://www.springer.com/series/4706}{Mathematical Studies}, \textbf{16}(2007) 27--70.

\bibitem{mienk} D. \href{http://www.renyi.hu/~gerbner}{Gerbner}, D. \href{http://www.cs.elte.hu/~dom}{P\'alv\"olgyi}, B. \href{http://www.renyi.hu/~patkos}{Patk\'os}, G. \href{http://www.cs.bme.hu/~wiener}{Wiener},
Finding the biggest and smallest element with one lie,
\href{http://www.elsevier.com/wps/find/journaldescription.cws_home/505609/description}{\textit{Discrete}}  \emph{Appl. Math.}, \textbf{158,} 9  (2010) 988--995.

\bibitem{HMOZ} M. \href{http://www.inf.ethz.ch/personal/hoffmann/}{Hoffmann}, J. \href{http://kam.mff.cuni.cz/~matousek/}{Matou\v sek}, Y. \href{http://www.is.titech.ac.jp/~okamoto/research/}{Okamoto}, Ph. \href{https://lists.inf.ethz.ch/pipermail/fag/2009-September/005384.html}{Zumstein}, Minimum and maximum against k lies, \href{http://arxiv.org/abs/1002.0562}{http://arxiv.org/abs/1002.0562}.

\bibitem{P}
A. Pelc, Searching games with errors -- Fifty years of
coping with liars, \href{http://www.elsevier.com/wps/find/journaldescription.cws_home/505625/description}{\emph{Theor. Comp. Sci.}} \textbf{270,} 1-2 (2002) 71--109.

\bibitem{Po}
I. \href{http://users.soe.ucsc.edu/~pohl/}{Pohl}, A sorting problem and its complexity,
\href{http://cacm.acm.org/}{\emph{Comm. ACM}} \textbf{15,} 6 (1972) 462--464.

\bibitem{RGL}
B. \href{http://ravi.cs.sonoma.edu/}{Ravikumar}, K. Ganesan, K. B. \href{http://www.acs.brockport.edu/~klakshma/}{Lakshmanan}, On selecting the
largest element in spite of erroneous information,  \emph{4th Annual Symposium on } \href{http://www.stacs-conf.org/}{\emph{Theoretical Aspects}} \emph{of Computer Science}, Passau, Germany,  Febr. 19--21, 1987.  \emph{Lecture Notes in Comp. Sci.} \textbf{247} (1987) pp. 88--99.


\bibitem{RGL}
B. Ravikumar, K. Ganesan, , On selecting the
largest element in spite of erroneous information, \href{http://www.stacs-conf.org/}{\emph{STACS}}, Ed. F. J. Brandenburg, G. Vidal-Naquet and M. Wirsing, 1987,
88--99.


\end{thebibliography}
\end{document}